\documentclass[5p,times]{elsarticle}

\usepackage{amsmath,amssymb,amsfonts}
\usepackage{amsthm}
\usepackage{algorithm}
\usepackage{algorithmic}
\usepackage{graphicx}
\usepackage{subcaption}
\usepackage{booktabs}
\usepackage{multirow}
\usepackage{makecell}
\usepackage{array,tabularx}
\usepackage{bm}
\usepackage[hidelinks]{hyperref}
\usepackage{lineno}
\usepackage{placeins}

\theoremstyle{plain}
\newtheorem{theorem}{Theorem}
\newtheorem{proposition}{Proposition}

\theoremstyle{definition}
\newtheorem{definition}{Definition}
\newtheorem{assumption}{Assumption}
\newtheorem{insight}{Insight}

\theoremstyle{remark}
\newtheorem{remark}{Remark}

\newcommand{\E}{\mathbb{E}}
\newcommand{\R}{\mathbb{R}}
\newcommand{\N}{\mathcal{N}}
\newcommand{\ext}{\mathrm{ext}}
\newcommand{\intr}{\mathrm{int}}
\newcommand{\rsq}{\mathrm{RSQ}}

\allowdisplaybreaks

\journal{Neurocomputing}

\begin{document}

\begin{frontmatter}

\title{Quality-Aware Exploration Budget Allocation for Cooperative Multi-Agent Reinforcement Learning}

\author[1]{Dahyun Oh}
\ead{qlass33@snu.ac.kr}
\author[1]{Minhyuk Yoon}
\ead{likenuclear@snu.ac.kr}
\author[1]{H. Jin Kim\corref{cor}}
\ead{hjinkim@snu.ac.kr}
\cortext[cor]{Corresponding author.}
\affiliation[1]{organization={Department of Aerospace Engineering, Seoul National University},
            city={Seoul},
            country={Republic of Korea}}

\begin{abstract}
Cooperative multi-agent reinforcement learning (MARL) requires agents to discover joint strategies in a combinatorially large state-action space, yet effective coordination configurations are exceedingly rare. Intrinsic motivation, which augments task rewards with novelty bonuses, is a popular approach for driving exploration, but its effectiveness hinges on the exploration intensity $\beta$, where too large a value overwhelms the task signal and causes coordination collapse, while too small a value prevents discovery of rare strategies. We address two complementary challenges: adapting $\beta$ globally over training, and allocating the exploration budget across agents whose intrinsic reward signals vary in reliability. Our framework combines a return-conditioned sigmoid schedule (RCB) for global intensity control with a per-agent Reward Signal Quality (RSQ) metric that concentrates the exploration budget on agents with reliable signals. The core insight is that agents receiving noisy intrinsic rewards should explore less aggressively, and this allocation can be determined automatically from signal-to-noise statistics. Successor Distance (SD), a quasimetric intrinsic reward, naturally produces distinguishable per-agent signal quality, completing the framework with convergence and ordering preservation guarantees. On seven cooperative benchmarks (MPE, SMAX, MABrax), our method achieves top-tier returns across all environments.
\end{abstract}


\end{frontmatter}

\section{Introduction}
\label{sec:introduction}

Cooperative multi-agent reinforcement learning (MARL) \\
trains teams of autonomous agents to jointly solve shared tasks through trial and error~\cite{albrecht2024multi}. A fundamental challenge is \emph{exploration}: the joint state-action space grows combinatorially with the number of agents, yet configurations yielding effective coordination are exceedingly rare~\cite{panait2005cooperative, hu2024review}. Without sufficient exploration, agents converge to suboptimal behaviors because they never discover the coordinated strategies that produce high team returns.

\begin{figure*}[!t]
    \centering
    \includegraphics[width=0.75\textwidth]{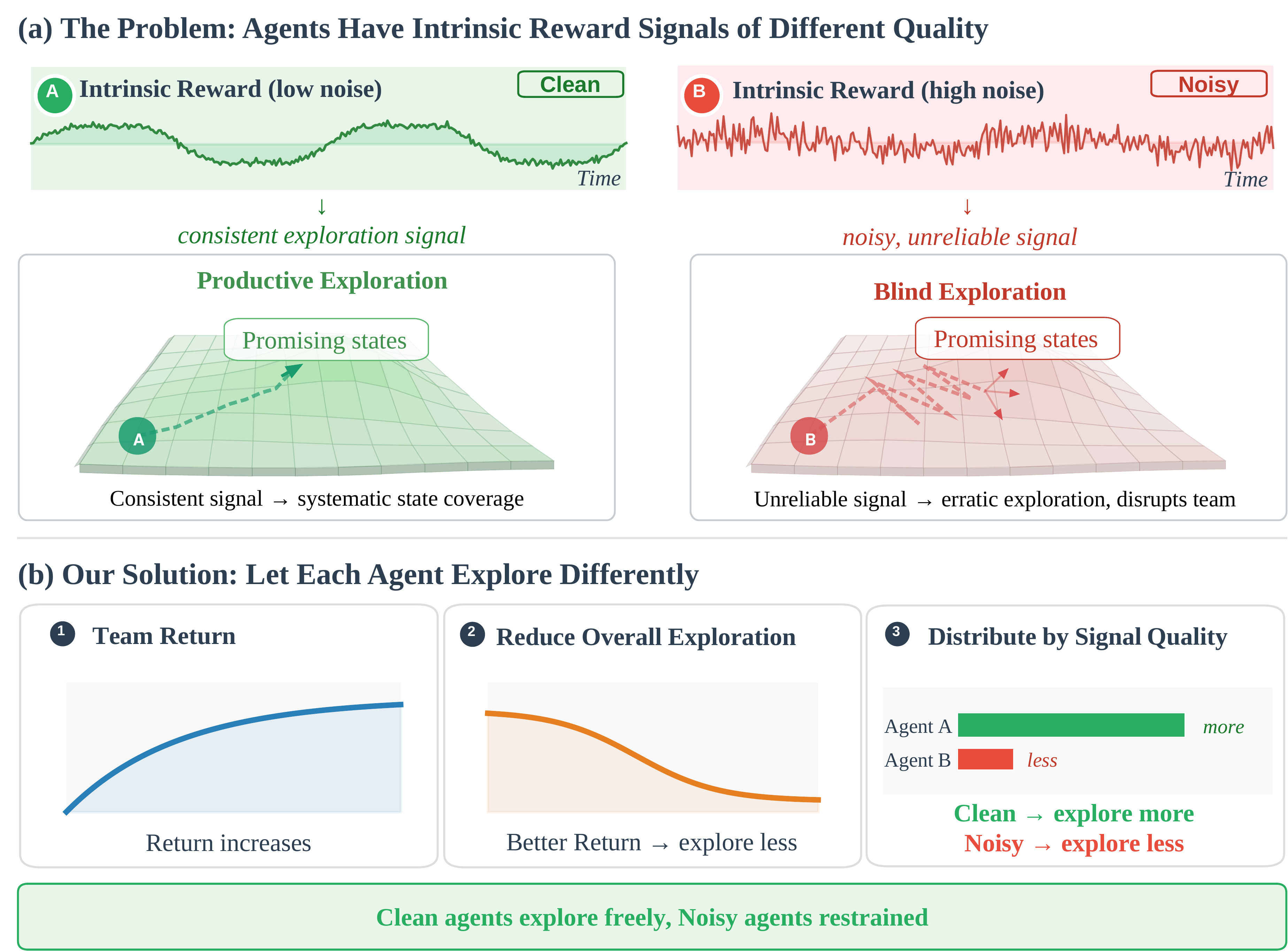}
    \caption{(a) Allocating exploration budget without considering signal quality lets noisy agents destabilize coordination. (b) Our framework adapts globally via RCB and per-agent via RSQ, concentrating the exploration budget on agents with reliable signals, which yields stronger and more stable returns than the tested baselines.}
    \label{fig:overview}
\end{figure*}

A wide variety of exploration-enhancing techniques have been proposed for cooperative MARL~\cite{gronauer2022multi}, including shared latent variables for coordinated behavior~\cite{mahajan2019maven}, constrained optimization of exploration budgets~\cite{chen2022redeeming}, and dual policy heads~\cite{li2023coin}, among many others~\cite{foerster2016learning, li2021celebrating, kim2023adaptive, du2019liir}. Among these diverse approaches, augmenting the task reward with intrinsic motivation bonuses has become the most widely adopted paradigm due to its generality~\cite{pathak2017curiosity, burda2018exploration, jaques2019social, zheng2021episodic}.

In the intrinsic motivation framework, the \emph{extrinsic reward} $r^{\ext}$ (the task-specific signal from the environment, e.g., reaching a goal or capturing a target) is supplemented with an \emph{intrinsic reward} $r_i^{\intr}$ that each agent generates internally to quantify the novelty of its experience, yielding the augmented reward $r_i = r^{\ext} + \beta \cdot r_i^{\intr}$ for each agent $i$. The \emph{exploration intensity} $\beta > 0$ governs how strongly exploration influences learning, and setting it correctly is the central difficulty. Existing methods either fix $\beta$ to a single value that must be tuned per task~\cite{mahajan2019maven, foerster2016learning, li2021celebrating, li2023coin} or adapt per-agent budgets based on task performance constraints~\cite{chen2022redeeming} without considering whether each agent's intrinsic reward is actually reliable. Even when the exploration budget is well-allocated for task feasibility, a subtler problem remains: different agents produce intrinsic rewards of different quality, and an agent whose novelty signal is dominated by noise can destabilize coordination regardless of how much budget it receives. Moreover, detecting these quality differences requires that the intrinsic reward itself produce signals of varying reliability across agents, which depends on how the intrinsic reward is computed (Sec.~\ref{subsec:sd}).

We identify two fundamental challenges that current methods leave unresolved. First, the optimal $\beta$ is task-dependent, and tasks that demand extensive exploration (e.g., navigating physical bottlenecks) require large $\beta$, while tasks where coordination emerges easily need small $\beta$ to avoid overwhelming the task signal. Manually tuning $\beta$ is impractical, especially since small changes can cause catastrophic training failures. Second, within a single task, different agents experience different intrinsic reward dynamics. An agent whose intrinsic reward signal is reliable can explore productively, while an agent with noisy or inconsistent intrinsic motivation risks destabilizing coordination through uncontrolled exploration. As illustrated in Fig.~\ref{fig:overview}(a), allocating exploration budget without accounting for intrinsic reward quality allows agents with noisy exploration signals to disrupt team coordination.

To address these challenges, we propose an integrated exploration framework whose three components each address a distinct requirement (Fig.~\ref{fig:overview}(b)). For global scheduling, we derive a Return-Conditioned $\beta$ (RCB) schedule, sigmoidal in team return, with provable convergence that adapts the exploration intensity to team learning progress. For per-agent allocation, we introduce the Reward Signal Quality (RSQ) metric, a computationally efficient affine modulation based on signal-to-noise ratio that preserves the allocation ordering of the information-optimal solution. For RSQ to differentiate agents, the intrinsic reward must produce signals whose reliability naturally varies across agents. Successor Distance (SD)~\cite{myers2024temporal}, which measures how far apart two states are in terms of the time it takes to travel between them, satisfies this condition. An agent that can explore on its own produces a consistent novelty signal across episodes, while an agent whose experience depends on what teammates happen to do produces a noisy, fluctuating signal. SD naturally reflects this difference, and RSQ detects it automatically. The three components address complementary situations. RCB actively modulates global exploration on tasks with wide return ranges, RSQ prevents training collapse by attenuating noisy agents on large-scale tasks where return variation is narrow, and SD provides the geometrically structured intrinsic signal that makes per-agent quality differentiation possible. Together, they achieve the highest or statistically tied return on all seven benchmarks.

Our contributions are as follows:
\begin{itemize}
\item A \textbf{Return-Conditioned $\beta$ (RCB)} for \\
global exploration intensity control with provable convergence under a contraction condition. RCB adapts $\beta$ to team learning progress, substantially reducing the need for per-task manual tuning (Sec.~\ref{subsec:rcb}).
\item A \textbf{Reward Signal Quality (RSQ) metric} for per-agent exploration budget allocation, motivated by information-theoretic analysis. The signal-to-noise ratio (SNR) emerges as the natural quality metric, and the affine RSQ modulation preserves the allocation ordering of the information-optimal solution. The key design guarantee is that RSQ enables exploration intensities that would otherwise cause training collapse, by selectively attenuating agents with unreliable signals (Sec.~\ref{subsec:rsq}).
\item A \textbf{complete framework} that integrates RCB, RSQ, and Successor Distance (SD) into a unified system. SD's quasimetric structure naturally produces distinguishable per-agent signal quality, a necessary condition for RSQ to differentiate agents. Replacing SD with entropy, RND, or count-based rewards eliminates this differentiation and causes training collapse (Sec.~\ref{subsec:sd}).
\item \textbf{Experiments on seven cooperative benchmarks} (two discrete MPE tasks, two SMAX combat tasks with 8 and 27 agents, three continuous MABrax tasks) with 10 seeds per configuration. Our method achieves top-tier returns across all environments. The source code is available at {\small\url{https://github.com/DH-O/RRS}}.
\end{itemize}

\section{Related Work}
\label{sec:related_work}

\subsection{Exploration in Single-Agent Reinforcement Learning}

Exploration in reinforcement learning has evolved from simple randomization to structured novelty-seeking. Classical approaches such as $\epsilon$-greedy~\cite{sutton2018reinforcement} and optimistic initialization~\cite{machado2015domain} add undirected randomness, while UCB methods~\cite{auer2002ucb} and posterior sampling~\cite{osband2016deep} provide principled uncertainty-driven exploration. In function-approximation settings, count-based methods~\cite{bellemare2016unifying} generalize visitation counts via density models, ICM~\cite{pathak2017curiosity} uses dynamics prediction error as a novelty proxy, RND~\cite{burda2018exploration} measures novelty through distillation error from a random target network, and random curiosity modules~\cite{li2020random} stabilize intrinsic rewards against catastrophic forgetting. Maximum entropy RL~\cite{haarnoja2018soft} encourages broad action distributions by incorporating policy entropy into the objective. These methods are designed for a single learner and do not account for multi-agent coordination, where one agent's exploration can disrupt another's learning~\cite{wang2020influence}.

\subsection{Exploration Paradigms in Cooperative MARL}

Exploration in cooperative MARL has been approached 
\\through diverse paradigms. Latent-variable methods~\cite{mahajan2019maven} condition agents on shared variables to produce diverse team behaviors, episodic memory~\cite{zheng2021episodic} provides joint-state novelty signals, diversity maximization~\cite{li2021celebrating} encourages behavioral variety via mutual information objectives, learned communication protocols~\cite{foerster2016learning, sukhbaatar2016learning, das2019tarmac} reduce partial observability through information sharing, and mutual-information-based teammate modeling~\cite{jiang2024macs} optimizes informative inter-agent communication. 
\\Maximum-entropy methods such as HASAC~\cite{liu2024hasac} encourage exploration via policy entropy regularization rather than reward augmentation, while COIN~\cite{li2023coin} maintains separate exploration and exploitation policy heads to decouple the two objectives.

Among these paradigms, augmenting the task reward with per-agent intrinsic motivation bonuses ($r_i = r^{\ext} + \beta \cdot r_i^{\intr}$) has emerged as the dominant framework due to its generality and compatibility with standard policy gradient architectures. Within this additive framework, various novelty signals have been proposed. Social influence~\cite{jaques2019social} and influence-based exploration~\cite{wang2020influence} reward agents for affecting teammates, while coordinated prediction errors~\cite{iqbal2019coordinated} and individualized intrinsic rewards~\cite{du2019liir} align bonuses with the team objective. Recent extensions include trajectory entropy maximization~\cite{li2025toward}, curiosity calibration via information bottleneck~\cite{cermic2025}, and sequential entropy accumulation~\cite{saeir2024}. Independently, Nimonkar et al.~\cite{nimonkar2025marl} show that contrastive distance functions naturally produce cooperative exploration, supporting the use of quasimetric distances as a basis for multi-agent intrinsic rewards.

Despite this progress, the \emph{form} of exploration (what novelty signal to use) has received far more attention than the \emph{intensity} of exploration (how much to explore). Nearly all additive methods above apply a fixed $\beta$ to every agent throughout training without considering whether each agent's intrinsic reward is reliable, leaving two questions unresolved: how should $\beta$ adapt over the course of learning, and how should the exploration budget be distributed across agents with different intrinsic reward quality? Our work addresses both questions, providing a global schedule that adapts to team learning progress and a per-agent allocation based on intrinsic reward signal quality.

\subsection{Adaptive Exploration and Constrained Approaches}

Adaptive methods tune exploration parameters online. In single-agent RL, meta-gradient methods~\cite{xu2018meta} and self-tuning actor-critics~\cite{zahavy2020self} adjust entropy coefficients automatically, 
\\SASR~\cite{ma2025sasr} adapts reward shaping via success-rate-driven Beta distributions, and maximum entropy exploration methods~\cite{li2020mee} separate target and explorer policies to handle deceptive rewards. In multi-agent settings, ADER~\cite{kim2023adaptive} learns per-agent entropy coefficients via value factorization. ADER is the closest prior work to per-agent exploration budget allocation and shares our motivation of agent-specific adaptation. However, it differs in three key aspects: (i) it modulates policy entropy rather than intrinsic reward intensity, which does not distinguish agents by intrinsic signal quality, (ii) it lacks an explicit optimality criterion for the resulting per-agent allocation, and (iii) it requires QMIX-style value decomposition (factoring the joint Q-value into per-agent components), making it incompatible with the MAPPO/IPPO architecture used by our baselines and most recent cooperative MARL work~\cite{yu2022mappo, rutherford2024jaxmarl}. EIPO~\cite{chen2022redeeming} treats intrinsic rewards as constraints rather than additive bonuses, preventing the exploration signal from overwhelming the task objective, and CIM~\cite{zheng2024cim} extends this line with constrained policy optimization for intrinsic reward coefficient tuning. These constrained approaches address \emph{whether} to explore (safe vs.\ unsafe) but not \emph{how to allocate} the exploration budget across agents whose signal quality varies. Our framework fills this gap. RCB provides global scheduling with provable convergence, while RSQ allocates exploration budgets across agents in proportion to their signal quality.

\subsection{Signal Quality and Information-Theoretic Perspectives}

The signal-to-noise ratio (SNR) is a fundamental measure of information reliability, quantifying how much useful signal a measurement contains relative to background noise~\cite{cover2006elements}. In reinforcement learning, policy gradient estimates are inherently noisy, and their SNR determines how reliably the gradient points toward improvement. Roberts and Tedrake~\cite{roberts2008snr} and Kuba et al.~\cite{kuba2021settling} analyze the SNR and variance of policy gradients as diagnostic tools, and noisy value estimates have been shown to degrade cooperative policies through overestimation~\cite{wu2022subavg}, but these works do not use signal quality to modulate exploration. Han et al.~\cite{han2026nsr} recently show that the noise-to-signal ratio of policy gradients increases sharply near optima, providing independent theoretical support for SNR-based modulation. From an information-theoretic perspective, information-directed sampling~\cite{russo2014learning} connects exploration to mutual information by allocating effort proportional to information gain, and Hsu et al.~\cite{hsu2024randomized} extend provably efficient randomized exploration to cooperative MARL with regret bounds, though without per-agent signal quality differentiation. Our work applies this reasoning at a different level: rather than designing the exploration mechanism, we optimize the \emph{allocation of exploration budget} across agents by modeling each agent's intrinsic reward as a noisy channel and applying SNR analysis to determine which agents can productively use their exploration budget~\cite{goldsmith2005wireless}.

\section{Background}
\label{sec:background}

\subsection{Decentralized Partially Observable MDP}

We model cooperative MARL as a Dec-POMDP~\cite{oliehoek2016decpomdp}, defined by the tuple $(\N,\, \mathcal{S},\, \{\mathcal{A}^i\},\, T,\, R,\, \{\Omega^i\},\, O,\, \gamma)$. Here $\N = \{1, \ldots, n\}$ is the agent set, $\mathcal{S}$ is the state space, $\mathcal{A}^i$ is agent $i$'s action space, $T\colon \mathcal{S} \times \mathcal{A} \to \Delta(\mathcal{S})$ is the transition function, $R\colon \mathcal{S} \times \mathcal{A} \to \R$ is the shared team reward, $\Omega^i$ is agent $i$'s observation space, $O$ is the observation function, and $\gamma \in [0,1)$ is the discount factor. We denote by $x_t^i \in \R^{d_x}$ the local state features extracted from agent $i$'s observation $o_t^i$ (e.g., spatial coordinates or joint velocities), which serve as input to per-agent Successor Distance encoders.

\subsection{Intrinsic Motivation for Exploration}

Intrinsic motivation supplements the task reward with exploration bonuses~\cite{bellemare2016unifying, pathak2017curiosity, burda2018exploration}, yielding the augmented per-agent reward
\begin{equation}
    r_i \;=\; r^{\ext} + \beta \cdot r_i^{\intr},
    \label{eq:intrinsic_reward}
\end{equation}
where $r_i^{\intr}$ is agent $i$'s intrinsic reward. The exploration intensity $\beta$ governs the trade-off between task performance and state-space coverage, and any $\beta > 0$ injects exploration pressure into the policy gradient. The central challenge is setting $\beta$, since too large a value overwhelms the task signal, too small a value yields insufficient coverage, and the optimal setting is both task-dependent and agent-dependent~\cite{chen2022redeeming}.

\subsection{Signal-to-Noise Ratio}

In communication theory~\cite{shannon1948mathematical, goldsmith2005wireless}, the signal-to-noise ratio $\mathrm{SNR} = \mu^2/\sigma^2$ quantifies signal reliability. Adaptive modulation adjusts transmission parameters based on channel SNR: high SNR permits aggressive transmission, low SNR requires conservative modulation~\cite{goldsmith1997variable}. We apply this principle to exploration management: high intrinsic reward signal quality permits aggressive exploration, low quality requires attenuated exploration.

\subsection{Successor Distance and Quasimetric Structure}
\label{subsec:sd_quasimetric}

Successor Distance (SD)~\cite{myers2024temporal} measures temporal reachability between states. Let $p_\gamma^\pi(s_f{=}y \mid s_0{=}x)$ denote the discounted successor measure, defined as $(1{-}\gamma)\sum_{t=0}^\infty \gamma^t \Pr^\pi(s_t{=}y \mid s_0{=}x)$.

\begin{definition}[Successor Distance]
\label{def:sd}
\begin{equation}
d_{\mathrm{SD}}^\pi(x, y) = \log \left(\frac{p_\gamma^\pi(s_f{=}y \mid s_0{=}y)}{p_\gamma^\pi(s_f{=}y \mid s_0{=}x)}\right).
\end{equation}
\end{definition}

SD satisfies the axioms of a \emph{quasimetric}~\cite{myers2024temporal},
\begin{itemize}
    \item identity, $d(x,y)=0 \iff x=y$,
    \item triangle inequality, $d(x,z) \leq d(x,y) + d(y,z)$,
    \item non-negativity, $d(x,y) \geq 0$,
\end{itemize}
while allowing asymmetry $d(x,y) \neq d(y,x)$ that captures directed reachability. This structure implies path-consistent exploration incentives that neither count-based nor RND methods provide.

In multi-agent settings, each agent $i$ maintains an encoder $\phi_i$ operating on local state features $x_t^i \in \R^{d_x}$ extracted from $o_t^i$, rather than the full observation. The intrinsic reward $r_{i,t}^{\intr} = \min_{k < t} d_{\phi_i}(x_k^i, x_t^i)$~\cite{jiang2025etd} encourages visits to states far from the agent's trajectory history.

\section{Methodology}
\label{sec:methodology}

\begin{figure*}[t]
    \centering
    \includegraphics[width=0.95\textwidth]{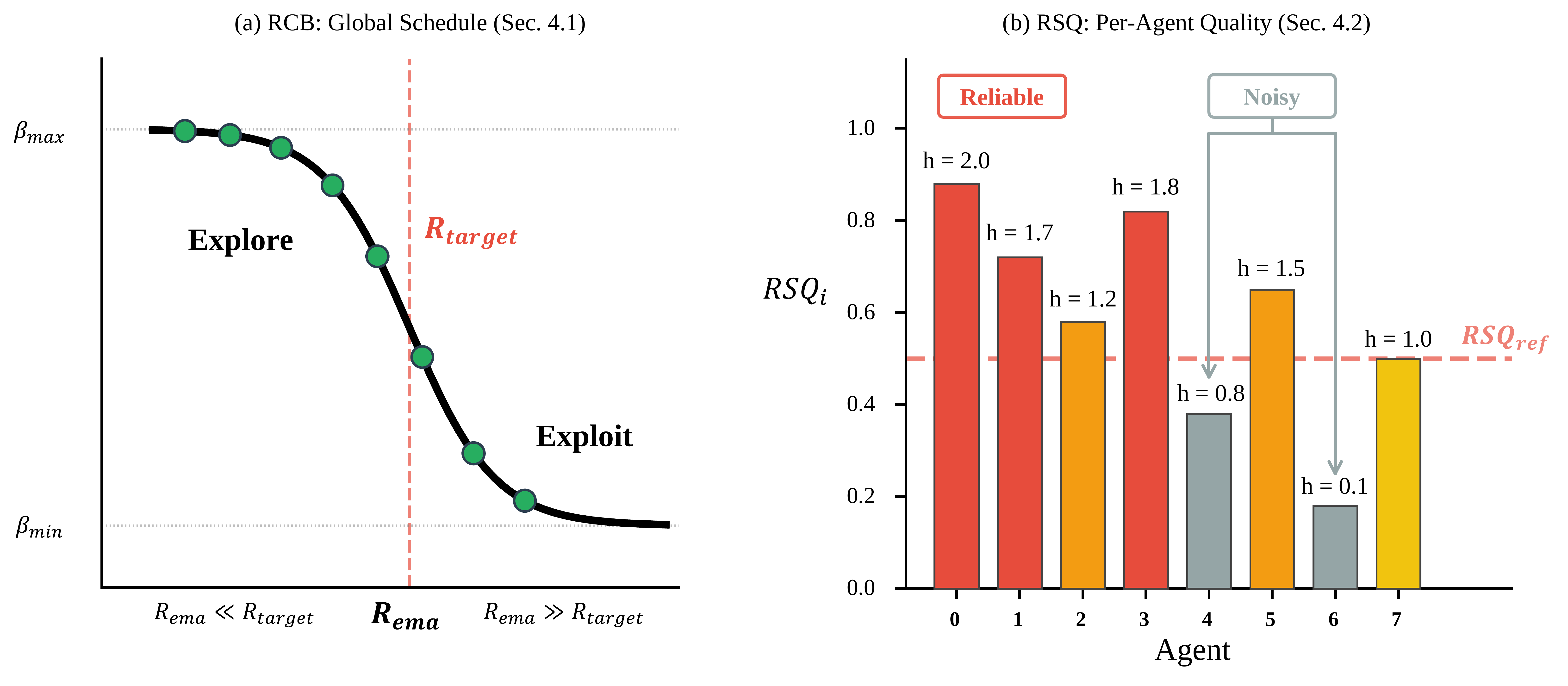}
    \caption{The two scheduling mechanisms of our framework. (a)~RCB adjusts global $\beta^{(k)}$ via a sigmoid of the team's EMA return: high $\beta$ early (explore), decaying as performance improves (exploit). (b)~RSQ measures per-agent intrinsic reward reliability and maps it to modulation weights $h_i$. Agents above $\mathrm{RSQ}_{\mathrm{ref}}$ receive amplified exploration ($h_i > 1$), while noisy agents are suppressed ($h_i < 1$).}
    \label{fig:architecture}
\end{figure*}

\subsection{Return-Conditioned Beta (RCB) Schedule}
\label{subsec:rcb}

A fixed $\beta$ in Eq.~\eqref{eq:intrinsic_reward} struggles in two distinct ways. \emph{Across tasks}, dense-reward environments need only a small bonus to nudge exploration, whereas sparse-reward environments require a much larger one to drive coverage at all, so no single value works for both. \emph{Within a single task}, agents benefit from aggressive exploration early in training when most of the state space is unfamiliar, but the same intensity becomes counterproductive later when the bonus starts to distract from refining a near-optimal policy. An ideal adaptive schedule must therefore satisfy three requirements, (1) the intensity must remain bounded, since unbounded $\beta$ overwhelms the task signal while zero $\beta$ stops exploration entirely, (2) it must decrease as team performance improves, so that agents shift from exploration to exploitation as they learn, and (3) the schedule must converge to a stable equilibrium rather than oscillating. We show that a sigmoid schedule driven by team return satisfies all three with provable convergence under a mild contraction condition.

\subsubsection{Team Return Tracking and Sigmoid Schedule}

Throughout this section we use two exploration parameters that index the current training iteration $k$. The first, $\beta^{(k)} > 0$, is simply the value of $\beta$ used at iteration $k$, a global exploration intensity shared by all agents and updated once per iteration (formally introduced in Definition~\ref{def:rcb} below). The second, $h_i^{(k)} > 0$, is a per-agent multiplicative weight that scales $\beta^{(k)}$ for agent $i$ based on the quality of its intrinsic reward signal at iteration $k$ (formally introduced in Definition~\ref{def:modulation} of Sec.~\ref{subsec:rsq}). The augmented reward in this section can simply be read as $r_i = r^{\ext} + \beta^{(k)} h_i^{(k)} r_i^{\intr}$.

Each training iteration $k$ then consists of collecting a batch of rollouts across $N_{\mathrm{envs}}$ parallel environments ($N_{\mathrm{steps}}$ timesteps each), computing $\beta^{(k)}$ and $h_i^{(k)}$ from the collected data, assembling the augmented reward, and performing a policy update. The exploration parameters are thus determined \emph{after} rollout collection and applied retroactively to the same batch for advantage computation, while action selection during rollout uses only the current policy, not $\beta$.

Let $R^{(k)}$ denote the mean team return over the completed episodes at iteration $k$. Because per-iteration returns fluctuate noisily, especially in sparse-reward tasks, the schedule should react to a slowly varying estimate of training progress rather than the raw $R^{(k)}$. We therefore track return progress via an exponential moving average (EMA), $R_{\mathrm{ema}}^{(k)}$, which collapses the noisy sequence $\{R^{(k)}\}$ into a smooth signal that the sigmoid in Definition~\ref{def:rcb} can act on,
\begin{equation}
    R_{\mathrm{ema}}^{(k+1)} = \alpha_R \cdot R^{(k)} + (1 - \alpha_R) \cdot R_{\mathrm{ema}}^{(k)},
    \label{eq:return_ema}
\end{equation}
where $\alpha_R \in (0,1)$ is the smoothing constant and we initialize $R_{\mathrm{ema}}^{(0)} = 0$. The reference return level $R_{\mathrm{target}}$ centers the sigmoid transition, and its robustness properties are discussed in Remark~\ref{rem:rtarget}.

\begin{definition}[Return-Conditioned Beta]
\label{def:rcb}
The return-conditioned exploration intensity at iteration $k$ is:
\begin{equation}
    \beta^{(k)} = \beta_{\min} + (\beta_{\max} - \beta_{\min}) \cdot \sigma\!\left(\kappa(R_{\mathrm{target}} - R_{\mathrm{ema}}^{(k)})\right),
    \label{eq:rcb}
\end{equation}
where $\sigma(x) = 1/(1+\exp(-x))$ is the logistic sigmoid, $\kappa > 0$ controls transition sharpness, and $\beta_{\min}, \beta_{\max} > 0$ are the intensity floor and ceiling.
\end{definition}

When $R_{\mathrm{ema}} \ll R_{\mathrm{target}}$, $\beta \approx \beta_{\max}$ (explore aggressively), and when $R_{\mathrm{ema}} \gg R_{\mathrm{target}}$, $\beta \approx \beta_{\min}$ (exploit). The sigmoid's bounded range, monotonicity in $R_{\mathrm{ema}}$, and centering at $R_{\mathrm{target}}$ satisfy all three design requirements by construction.

\begin{remark}[Robustness of $R_{\mathrm{target}}$ and $\kappa$]
\label{rem:rtarget}
$R_{\mathrm{target}}$ and $\kappa$ are not precision-sensitive hyperparameters. Because the logistic sigmoid saturates smoothly, $R_{\mathrm{target}}$ only needs to lie within the range of returns achievable during training, and $\kappa$ only needs to be small enough that the explore-to-exploit transition occurs gradually over that range. In practice, both can be set from coarse task-level information (e.g., approximate return scale), analogous to target entropy $-\log|\mathcal{A}|$ in SAC~\cite{haarnoja2018soft}. We provide a formal derivation of the sigmoid's transition bandwidth and an empirical sensitivity sweep in Sec.~\ref{subsec:rtarget_sensitivity}.
\end{remark}

\subsubsection{RCB Convergence}

\begin{assumption}[Return Response]
\label{assum:return_response}
For a fixed exploration intensity $\beta$, suppose the policy is trained long enough to reach approximate equilibrium. Let $\bar{R}(\beta)$ denote the resulting steady-state expected team return. We assume $\bar{R}(\beta)$ is continuous, bounded in $[R_{\min}, R_{\max}]$, and that the observed return at iteration $k$ satisfies $R^{(k)} = \bar{R}(\beta^{(k)}) + \xi^{(k)}$, where $\xi^{(k)}$ is zero-mean noise with bounded variance $\sigma_\xi^2$.
\end{assumption}

This is a standard assumption in adaptive learning rate analysis~\cite{borkar2008stochastic}: the exploration parameter ($\beta$) changes slowly enough relative to the policy that the policy approximately converges before $\beta$ updates again.

The key idea behind the following theorem is that the RCB update forms a feedback loop ($\beta \to \text{return} \to R_{\mathrm{ema}} \to \beta$), and if each step of this loop shrinks differences (a \emph{contraction}), then the loop converges to a unique equilibrium.

\begin{theorem}[RCB Schedule Convergence]
\label{thm:rcb_convergence}
Under Assumption~\ref{assum:return_response}, if the contraction condition
\begin{equation}
    \kappa(\beta_{\max} - \beta_{\min}) \cdot \sup_\beta |\bar{R}'(\beta)| \cdot \tfrac{1}{4} < 1
    \label{eq:contraction}
\end{equation}
holds (the $1/4$ comes from the maximum slope of the logistic sigmoid), then:
\begin{itemize}
    \item[(1)] \textbf{Uniqueness}: the equilibrium $(R^*, \beta^*)$ is unique.
    \item[(2)] \textbf{Deterministic convergence}: the tracking error contracts geometrically,
    \begin{equation}
        |R_{\mathrm{ema}}^{(k)} - R^*| \leq \rho^k\, |R_{\mathrm{ema}}^{(0)} - R^*|, \quad \rho = 1 - \alpha_R(1 - L_\Phi) \in (0,1),
        \label{eq:det_convergence}
    \end{equation}
    where $L_\Phi \triangleq \sup_\beta |\Phi'(\beta)|$ is the Lipschitz constant of the RCB map $\Phi$ (derived in \ref{app:rcb_convergence}).
    \item[(3)] \textbf{Stochastic convergence}: when observed returns are noisy ($\sigma_\xi^2 > 0$), $R_{\mathrm{ema}}$ satisfies
    \begin{equation}
        \E\bigl[|R_{\mathrm{ema}}^{(k)} - R^*|^2\bigr] \;\to\; O\!\left(\frac{\alpha_R \,\sigma_\xi^2}{1 - L_\Phi}\right).
        \label{eq:stoch_convergence}
    \end{equation}
\end{itemize}
\end{theorem}

\begin{proof}
The proof applies Banach's contraction mapping theorem~\cite{granas2003fixed} to the composed map
\[
\textstyle \Phi(\beta) = \beta_{\min} + (\beta_{\max} - \beta_{\min})\,\sigma\!\bigl(\kappa(R_{\mathrm{target}} - \bar{R}(\beta))\bigr),
\]
then bounds the stochastic residual via a linear recursion. See \ref{app:rcb_convergence} for the proof sketch.
\end{proof}

\begin{remark}[What Theorem~\ref{thm:rcb_convergence} delivers]
\label{rem:beyond_banach}
Theorem~\ref{thm:rcb_convergence} differs from a generic contraction-mapping argument in two specific ways. First, the Lipschitz constant is obtained in closed form, $L_\Phi = \kappa(\beta_{\max}-\beta_{\min})\sup_\beta|\bar{R}'(\beta)|/4$, by combining the chain rule with the analytic slope bound $\sigma'(\cdot) \leq 1/4$, rather than left as an abstract symbol. Second, part~(3) is not a fixed-point result but an $L^2$ recursion, in the stochastic-approximation style of~\cite{borkar2008stochastic}, that handles the noisy iteration $R^{(k)} = \bar{R}(\beta^{(k)}) + \xi^{(k)}$. These two features are what allow Theorem~\ref{thm:rcb_convergence} to deliver, beyond a generic Banach citation, a verifiable hyperparameter inequality on $(\kappa, \beta_{\min}, \beta_{\max})$ that the practitioner can check before training and an explicit noise-floor bound $O(\alpha_R\sigma_\xi^2/(1-L_\Phi))$ tying the EMA smoothing constant $\alpha_R$ to the asymptotic tracking error.
\end{remark}

In practice, the contraction condition is easily satisfied. With $\kappa = 0.01$ for all experiments, the left-hand product of Eq.~\eqref{eq:contraction} stays well below 1, and we confirm this empirically in Sec.~\ref{subsec:ablation}.

\subsection{Reward Signal Quality (RSQ)}
\label{subsec:rsq}

RCB provides a global exploration intensity $\beta^{(k)}$, and the remaining question is how to distribute this budget across $n$ different agents. Each agent's intrinsic reward contains both useful signal and noise, much like a communication channel. In communication theory, the classical \emph{water-filling} solution~\cite{cover2006elements, goldsmith2005wireless} distributes transmission power across parallel channels in proportion to their signal quality. We apply the same principle to exploration: agents with reliable intrinsic rewards (high SNR) receive more exploration budget, while noisy agents receive less. This analysis motivates the computationally efficient RSQ metric, a simple affine approximation that preserves the SNR-based ordering of the optimal allocation.

\subsubsection{Multi-Agent Exploration as Information Maximization}

Under uniform allocation, all agents explore with the same intensity $\beta^{(k)}$, regardless of whether their intrinsic reward reliably indicates productive exploration or is dominated by noise. A better strategy concentrates the budget on agents with high-quality signals (Fig.~\ref{fig:architecture}b).

We formalize this intuition by modeling each agent $i$'s intrinsic reward as a noisy observation: a useful signal $f_i$ (e.g., how much novel territory the agent covers) plus zero-mean noise $\epsilon_i$ from stochastic transitions and other agents' actions. The ratio of signal power to noise power defines a signal-to-noise ratio $\mathrm{SNR}_i$ for each agent (we make this precise in Sec.~\ref{subsec:rsq_metric}).

To quantify the useful information in each agent's intrinsic reward, we model each agent as a noisy channel. Among all distributions with a given mean and variance, the Gaussian carries the least mutual information~\cite{cover2006elements}. Any allocation that is optimal under the Gaussian model therefore remains valid (though possibly suboptimal) for the true distribution, giving a conservative lower bound. Under this model, the mutual information between agent $i$'s exploration and its intrinsic reward is:
\begin{equation}
    I_i = \tfrac{1}{2}\log(1 + \beta_i^2 \cdot \mathrm{SNR}_i).
\end{equation}
High $I_i$ means the intrinsic reward reliably distinguishes informative actions from uninformative ones.

The goal is to choose per-agent intensities $\beta_i$ that maximize the total information gain $\sum_i I_i$ subject to a budget constraint. Since the augmented reward is $r_i = r^{\ext} + \beta_i r_i^{\intr}$, the variance of the intrinsic component in the policy gradient scales as $\beta_i^2 \mathrm{Var}(r_i^{\intr})$, making $\beta_i^2$ the natural measure of each agent's exploration impact. The total impact is $\sum_i \beta_i^2 = n(\beta^{(k)})^2 \triangleq B$. Maximizing $\sum_i I_i$ subject to $\sum_i \beta_i^2 \leq B$ is the classical \emph{power allocation} problem, whose solution is the standard water-filling strategy~\cite[Ch.~10]{cover2006elements}. Identifying $p_i = \beta_i^2$ as the allocated power for agent~$i$,
\begin{equation}
    p_i^{*,\mathrm{WF}} = \left(\nu - \frac{1}{\mathrm{SNR}_i}\right)^+, \quad \beta_i^{*,\mathrm{WF}} = \sqrt{p_i^{*,\mathrm{WF}}},
    \label{eq:water_filling}
\end{equation}
where $(x)^+ = \max(x, 0)$ and the water level $\nu$ is chosen so $\sum_i p_i^{*,\mathrm{WF}} = B$. Agents with $1/\mathrm{SNR}_i \geq \nu$ receive zero allocation. We subsequently express the per-agent intensity directly as $\beta_i = \beta \cdot h_i$, so that $p_i = (\beta h_i)^2$.

\subsubsection{Reward Signal Quality: Computing SNR in Practice}
\label{subsec:rsq_metric}

The water-filling allocation~\eqref{eq:water_filling} requires a per-agent $\mathrm{SNR}_i$ to implement. We now make the abstract $\mathrm{SNR}_i$ concrete. At each training iteration, agent $i$ produces $N_{\mathrm{steps}} \times N_{\mathrm{envs}}$ intrinsic reward samples across all timesteps and parallel environments. The batch mean reflects the consistent signal strength, while the batch variance captures unpredictable fluctuations. We track these via exponential moving averages (EMAs), which smooth out iteration-to-iteration noise:
\begin{align}
    \mu_i^{(k+1)} &= \alpha \cdot \bar{r}_i^{\intr,(k)} + (1-\alpha) \cdot \mu_i^{(k)}, \label{eq:ema_mu} \\
    (\sigma_i^{(k+1)})^2 &= \alpha \cdot \widehat{\mathrm{Var}}(r_i^{\intr,(k)}) + (1-\alpha) \cdot (\sigma_i^{(k)})^2, \label{eq:ema_sigma}
\end{align}
where $\bar{r}_i^{\intr,(k)}$ and $\widehat{\mathrm{Var}}(r_i^{\intr,(k)})$ are the batch mean and variance of agent $i$'s intrinsic reward at iteration $k$. The Gaussian assumption is justified by the central limit theorem. The signal-to-noise ratio is then $\mathrm{SNR}_i = \mu_i^2 / \sigma_i^2$, following the standard power decomposition $\E[X^2] = \mu_X^2 + \sigma_X^2$~\cite{cover2006elements}.

\begin{definition}[Reward Signal Quality]
\label{def:rsq}
For agent $i$ at iteration $k$:
\begin{equation}
    \rsq_i^{(k)} = \frac{(\mu_i^{(k)})^2}{(\mu_i^{(k)})^2 + (\sigma_i^{(k)})^2 + \epsilon},
    \label{eq:rsq}
\end{equation}
where we apply the transform $x \mapsto x/(1+x)$ to $\mathrm{SNR}_i$, mapping it to $[0,1]$ so that $\operatorname{RSQ}=1$ means a perfect signal. We set $\epsilon = 10^{-8}$ for numerical stability.
\end{definition}

By construction, RSQ is invariant to positive rescaling of intrinsic rewards, eliminating per-task calibration.

\subsubsection{Per-Agent Exploration Intensity via Affine Modulation}

\begin{definition}[Affine RSQ Modulation]
\label{def:modulation}
The per-agent exploration intensity combines global RCB with individual RSQ:
\begin{equation}
    \beta_i^{(k)} = \beta^{(k)} \cdot h(\rsq_i^{(k)}),
    \label{eq:two_level}
\end{equation}
where the modulation function $h: [0,1] \to \R_{>0}$ is defined as:
\begin{equation}
    h(\rsq_i) = \mathrm{clip}\!\left(1 + \lambda(\rsq_i - \rsq_{\mathrm{ref}}),\; h_{\min},\; h_{\max}\right),
    \label{eq:affine_mod}
\end{equation}
with $\lambda > 0$ controlling sensitivity, $\rsq_{\mathrm{ref}}$ the reference RSQ level, and $[h_{\min}, h_{\max}]$ bounding the modulation range.
\end{definition}

The modulation function has three regimes: agents with RSQ above $\rsq_{\mathrm{ref}}$ receive amplified exploration ($h > 1$), those at the reference level are unchanged ($h = 1$), and those with low RSQ are suppressed ($h < 1$). The clipping bounds $[h_{\min}, h_{\max}]$ prevent extreme amplification or complete suppression. When all agents have similar intrinsic reward quality, $h_i \approx 1$ for all, recovering pure RCB. The specific parameter values are listed in Table~\ref{tab:hyperparams}.

\begin{remark}[Default Values for $\rsq_{\mathrm{ref}}$, $h_{\min}$, and $h_{\max}$]
\label{rem:modulation_defaults}
These three parameters have natural default values that require less per-task tuning than the remaining hyperparameters. The reference level $\rsq_{\mathrm{ref}} = 0.5$ corresponds to $\mathrm{SNR} = 1$ (equal signal and noise power), the natural boundary between ``reliable'' and ``noisy'' agents. The upper bound $h_{\max} = 2.0$ caps amplification at $2\times$ the global intensity, preventing any single agent from dominating the exploration budget. The lower bound $h_{\min} = 0.1$ ensures that even the noisiest agents retain at least $10\%$ of the global intensity, avoiding the irreversible suppression problem that afflicts exact water-filling (\ref{app:wf_comparison}). These values are shared across all seven environments (Table~\ref{tab:hyperparams}).
\end{remark}

\subsubsection{Why Affine Modulation Instead of Exact Water-Filling}

The water-filling solution~\eqref{eq:water_filling} is information-theoretically optimal but has a practical weakness: it assigns $\beta_i = 0$ to agents whose estimated SNR falls below the water level, permanently suppressing their exploration based on potentially noisy EMA estimates. The affine RSQ modulation (Definition~\ref{def:modulation}) replaces exact water-filling with a single-line formula that clips allocations to $[h_{\min}, h_{\max}]$, so that every agent retains a minimum exploration intensity. The following proposition confirms that this simple approximation allocates in exactly the same order as the optimal solution: agents with higher SNR always receive more exploration budget. \ref{app:wf_comparison} validates this design choice empirically: exact water-filling causes training collapse on corridor because noisy SNR estimates permanently suppress agents, while the affine modulation achieves stable coordination.

\begin{proposition}[Allocation Ordering Preservation]
\label{thm:ordering}
Let $n$ agents have signal-to-noise ratios $\mathrm{SNR}_1 \geq \mathrm{SNR}_2 \geq \cdots \geq \mathrm{SNR}_n > 0$. Then:
\begin{enumerate}
    \item[(i)] The water-filling allocation satisfies $p_1^{\mathrm{WF}} \geq p_2^{\mathrm{WF}} \geq \cdots \geq p_n^{\mathrm{WF}}$.
    \item[(ii)] The affine RSQ modulation satisfies $h_1 \geq h_2 \geq \cdots \geq h_n$, yielding per-agent intensities $\beta_1 \geq \beta_2 \geq \cdots \geq \beta_n$.
    \item[(iii)] For any pair $(i,j)$ with $\mathrm{SNR}_i > \mathrm{SNR}_j$, if both modulation weights are unsaturated ($h_{\min} < h_j \leq h_i < h_{\max}$), then the ordering is strict: $h_i > h_j$.
\end{enumerate}
\end{proposition}

\begin{proof}[Proof sketch]
The composition of the strictly increasing $\rsq(\mathrm{SNR}) = \mathrm{SNR}/(1+\mathrm{SNR})$ with the non-decreasing clipped affine $h$ is non-decreasing in $\mathrm{SNR}$, with strict increase whenever $h_i, h_j \in (h_{\min}, h_{\max})$. Part~(i) inherits the same monotonicity from~\eqref{eq:water_filling}, since $1/\mathrm{SNR}_i$ is decreasing in $\mathrm{SNR}_i$.
\end{proof}

Empirically, the affine modulation outperforms exact water-filling on corridor with substantially lower variance and higher mean return (Table~\ref{tab:wf_comparison}). See \ref{app:wf_comparison} for the full comparison.

\begin{algorithm}[t]
\caption{Our framework: RCB schedule and RSQ modulation over per-agent Successor Distance}
\label{alg:framework}
\begin{algorithmic}[1]
\REQUIRE Policies $\{\pi^i\}$, critic $V$, SD encoders $\{\phi_i\}$
\REQUIRE $\beta_{\min}, \beta_{\max}, \kappa, \alpha_R, \lambda, \rsq_{\mathrm{ref}}, h_{\min}, h_{\max}$
\STATE Initialize: $R_{\mathrm{ema}} \gets 0$, $\mu_i \gets 0$, $\sigma_i^2 \gets 1$
\WHILE{not converged}
    \STATE Sample rollouts across $E$ parallel environments
    \STATE $R^{(k)} \gets$ mean team return from completed episodes
    \STATE \textbf{// Phase 1: Global RCB Schedule}
    \STATE $R_{\mathrm{ema}} \gets \alpha_R R^{(k)} + (1-\alpha_R) R_{\mathrm{ema}}$
    \STATE $\beta \gets \beta_{\min} + (\beta_{\max}-\beta_{\min})\sigma(\kappa(R_{\mathrm{target}} - R_{\mathrm{ema}}))$
    \STATE \textbf{// Phase 2: Per-Agent RSQ Modulation}
    \STATE Compute $r_{i,t}^{\intr}$ via SD for each agent $i$
    \FOR{each agent $i$}
        \STATE Update $\mu_i, \sigma_i^2$ from $r_i^{\intr}$ via EMA~\eqref{eq:ema_mu}--\eqref{eq:ema_sigma}
        \STATE $\rsq_i \gets \mu_i^2/(\mu_i^2 + \sigma_i^2 + \epsilon)$
        \STATE $h_i \gets \mathrm{clip}(1 + \lambda(\rsq_i - \rsq_{\mathrm{ref}}), h_{\min}, h_{\max})$
    \ENDFOR
    \STATE \textbf{// Phase 3: Policy Update}
    \STATE $r_{i,t} \gets r_t^{\ext} + \beta \cdot h_i \cdot r_{i,t}^{\intr}$
    \STATE Compute advantages $\hat{A}_{i,t}$ via GAE, update via PPO
    \STATE \textbf{// Phase 4: Update SD Encoders}
    \STATE Update SD encoders $\{\phi_i\}$ via contrastive loss
\ENDWHILE
\end{algorithmic}
\end{algorithm}

\subsection{Successor Distance under RSQ Modulation}
\label{subsec:sd}

Following~\cite{myers2024temporal}, we use per-agent Successor Distance (SD) for intrinsic motivation. Each agent $i$ maintains an encoder $\phi_i$ trained via symmetric InfoNCE~\cite{oord2018representation} on local state features $x_t^i \in \R^{d_x}$ (e.g., spatial coordinates for grid worlds, joint velocities for continuous control). The intrinsic reward $r_{i,t}^{\intr} = \min_{k < t}\, d_{\phi_i}(x_k^i, x_t^i)$~\cite{jiang2025etd} encourages visits to states far from the agent's trajectory history.

We use per-agent SD rather than joint SD because it avoids the exponential sample complexity of learning distances in the joint state space~\cite{poole2019variational}, and because SD produces intrinsic rewards that differ substantially across agents depending on their spatial roles. An agent exploring open space consistently receives large SD values (high $\mu_i$, low $\sigma_i$, high RSQ), while an agent stuck near obstacles receives small, erratic values (low $\mu_i$, high $\sigma_i$, low RSQ). This per-agent variation is what enables RSQ to assign meaningfully different $h_i$ values (Assumption~\ref{assum:quality_variation}).

\begin{assumption}[Per-Agent Signal Quality Variation]
\label{assum:quality_variation}
There exists a minimum pairwise RSQ gap $\Delta_{\min} > 0$ such that, for at least one pair of agents $(i, j)$, $|\rsq_i - \rsq_j| \geq \Delta_{\min}$ holds during the active exploration phase. That is, the intrinsic reward method produces distinguishable per-agent signal quality.
\end{assumption}

When $\Delta_{\min} = 0$, all modulation weights collapse to $h_i \approx 1$ and the method reduces to RCB-only.

\begin{insight}[SD Satisfies Assumption~\ref{assum:quality_variation} Empirically]
\label{insight:quasimetric}
On corridor, SD produces inter-agent RSQ gaps exceeding $0.3$ (Fig.~\ref{fig:rsq_dynamics}). Replacing SD with entropy, RND, or count-based rewards yields near-uniform RSQ ($h_i \approx 1$), violating Assumption~\ref{assum:quality_variation} and yielding strongly negative mean return (\ref{app:intrinsic_source}, Table~\ref{tab:intrinsic_source}).
\end{insight}

\subsection{Complete Training Algorithm}
\label{subsec:algorithm}

Algorithm~\ref{alg:framework} presents the unified framework.
The EMA initialization ($\mu_i = 0$, $\sigma_i^2 = 1$) provides a natural warmup. Initially $\rsq_i \approx 0$ for all agents, yielding $h_i = h_{\min}$, so the effective per-agent intensity is $\beta_{\max} \cdot h_{\min} = 0.05$ regardless of the global $\beta$. As intrinsic reward statistics accumulate over the first few hundred updates, RSQ values rise and differentiation emerges, gradually increasing the exploration budget toward its intended level.

Six of ten framework hyperparameters are shared across all environments. The remaining four are adapted per domain from 2--3 candidates on a single seed. Table~\ref{tab:hyperparams} lists all values.

\begin{table}[!t]
    \centering
    \caption{Hyperparameters of our framework per environment. Four parameters ($\alpha_R$, $\alpha$, $\rsq_{\mathrm{ref}}$, $h_{\max}$) are shared across all environments. The domain-adapted parameters were selected as follows: $R_{\mathrm{target}}$ is set to a reasonable task return (robust over $20\times$ range, Table~\ref{tab:sensitivity_rtarget}), and the remaining parameters were chosen from 2--3 candidates on a single seed per environment.}
    \label{tab:hyperparams}
    \footnotesize
    \setlength{\tabcolsep}{3.5pt}
    \begin{tabular}{lccccccc}
        \toprule
        Parameter & Corridor & Tag & 3s5z & 27m & Ant & Ant-ball & HC \\
        \midrule
        \multicolumn{8}{l}{\textit{Shared across all environments}} \\
        $\alpha_R$ (return EMA) & \multicolumn{7}{c}{$0.03$} \\
        $\alpha$ (RSQ EMA) & \multicolumn{7}{c}{$0.1$} \\
        $\rsq_{\mathrm{ref}}$ & \multicolumn{7}{c}{$0.5$} \\
        $h_{\max}$ & \multicolumn{7}{c}{$2.0$} \\
        \midrule
        \multicolumn{8}{l}{\textit{Adapted per domain}} \\
        $\beta_{\max}$ & $0.5$ & $0.5$ & $0.2$ & $0.5$ & $0.5$ & $0.5$ & $0.5$ \\
        $\kappa$ & $0.01$ & $0.015$ & $0.01$ & $0.01$ & $0.01$ & $0.01$ & $0.01$ \\
        $\beta_{\min}$ & $0.1$ & $0.15$ & $0.05$ & $0.3$ & $0.05$ & $0.05$ & $0.05$ \\
        $R_{\mathrm{target}}$ & $400$ & $400$ & $1.5$ & $3.0$ & $500$ & $500$ & $500$ \\
        $\lambda$ & $3.0$ & $2.0$ & $2.0$ & $5.0$ & $2.0$ & $2.0$ & $2.0$ \\
        $h_{\min}$ & $0.1$ & $0.1$ & $0.8$ & $0.1$ & $0.1$ & $0.1$ & $0.1$ \\
        \bottomrule
    \end{tabular}
\end{table}

\section{Experiments}
\label{sec:experiments}

We evaluate our method on seven cooperative benchmarks: two MPE tasks~\cite{lowe2017maddpg}, two SMAX combat tasks, and three MABrax locomotion tasks, all implemented in JaxMARL~\cite{rutherford2024jaxmarl}.

\subsection{Setup}

\subsubsection{Environments}

\begin{figure*}[!t]
    \centering
    \includegraphics[width=\textwidth]{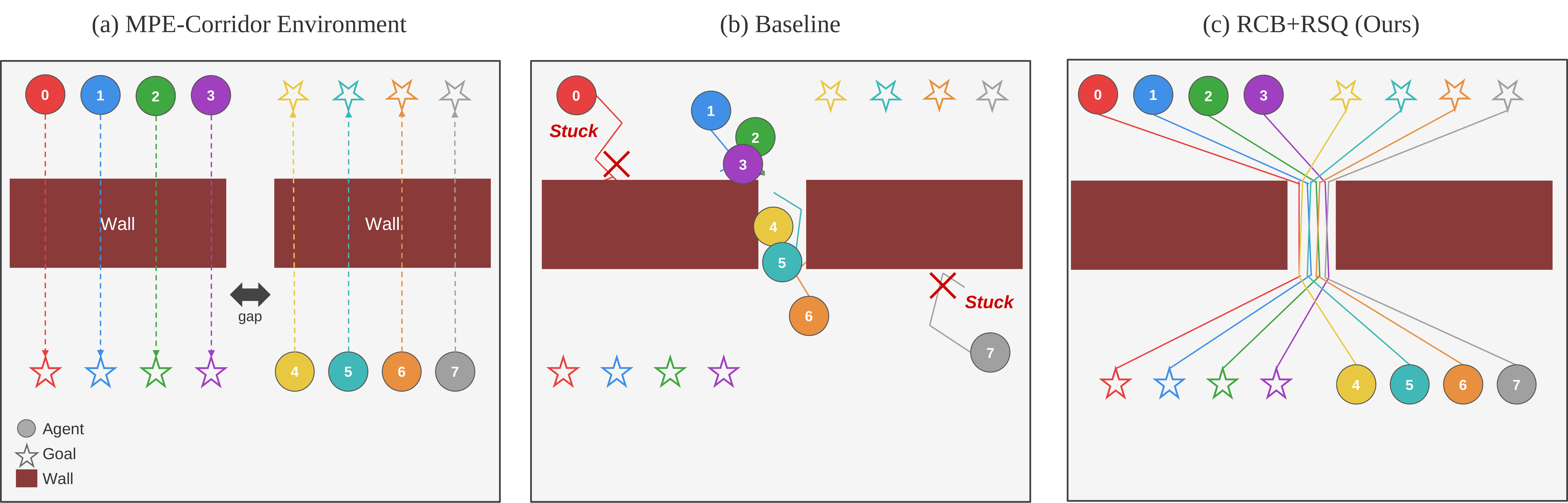}
    \caption{MPE-corridor environment. (a) Environment layout: 8 agents (4 top, 4 bottom) must navigate through a narrow bottleneck (width 0.8) to reach goals on the opposite side. (b) Baseline methods that do not account for per-agent signal quality fail to coordinate, with edge agents getting stuck while center agents cause congestion at the bottleneck. (c) Our method achieves coordinated passage by attenuating noisy agents' exploration and adapting the global exploration intensity to team progress.}
    \label{fig:corridor_env}
\end{figure*}

Figs.~\ref{fig:corridor_env} and~\ref{fig:env_showcase} together illustrate the seven evaluation environments. We use two \textbf{MPE} tasks~\cite{lowe2017maddpg}, namely \emph{corridor} (8 agents, narrow bottleneck, Fig.~\ref{fig:corridor_env}) and \emph{tag} (6 predators vs.\ 2 scripted prey, Fig.~\ref{fig:env_showcase}a). Two \textbf{SMAX} tasks~\cite{rutherford2024jaxmarl}, a JAX reimplementation of SMAC~\cite{samvelyan2019starcraft} with simplified combat mechanics: \emph{3s5z} (8 agents vs.\ 9 opponents, heterogeneous unit types) and \emph{27m} (27 marines vs.\ 30 opponents). SMAX returns are not directly comparable to SMAC. Following recent evaluations~\cite{mahjoub2025sable, tessera2025hypermarl}, we report mean episode return rather than binary win rate. Three \textbf{MABrax}~\cite{rutherford2024jaxmarl,freeman2021brax} continuous locomotion tasks: \emph{ant\_4x2} (4-segment ant, each segment is an agent with 2 actuators), \emph{ant\_ball} (same ant carrying a ball on a tray with partial observability adapted from~\cite{jang2025eigensafe}), and \emph{halfcheetah\_6x1} (6-segment cheetah, 1 actuator each). On ant\_ball, the forward reward is amplified ($3\times$ weight) and each agent observes only a neighbor's joint velocities in a directed ring (24 of 46 dimensions), making coordinated gaits essential. Environment details (observation spaces, reward structure, episode lengths) are in \ref{app:implementation}.

\begin{figure*}[!t]
    \centering
    \includegraphics[width=\textwidth]{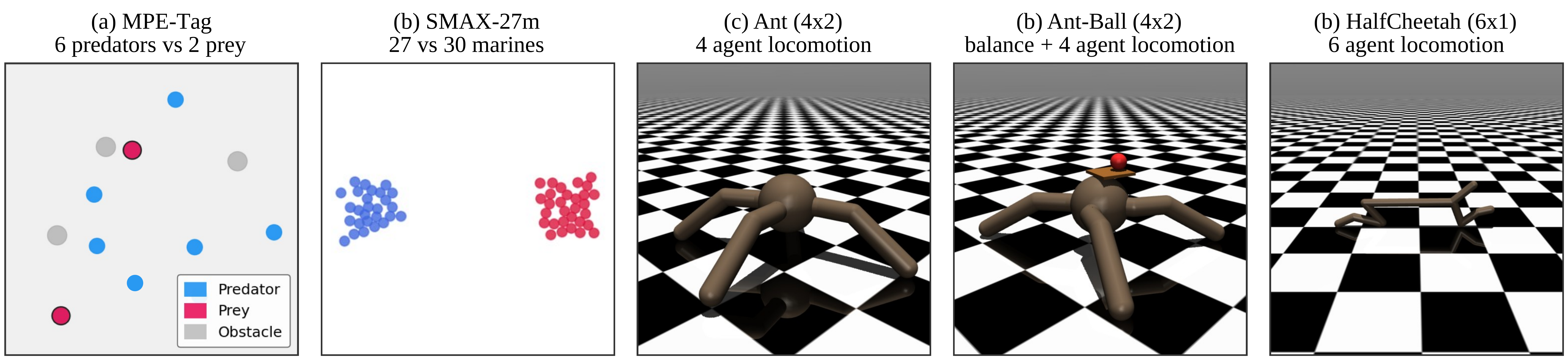}
    \caption{Evaluation environments (excluding corridor, shown in Fig.~\ref{fig:corridor_env}). (a)~MPE-Tag: 6 predators (blue) chase 2 scripted prey (red). (b)~SMAX-27m: 27 allied marines (blue) vs.\ 30 enemies (red). (c)~MABrax Ant (\emph{ant\_4x2}): 4 segment ant with 2 actuators per segment, where each segment is an independent agent. (d)~MABrax Ant-Ball: same ant carrying a ball (red) on a tray while walking forward. (e)~MABrax HalfCheetah (\emph{halfcheetah\_6x1}): 6 segment locomotion with 1 actuator per segment.}
    \label{fig:env_showcase}
\end{figure*}

These environments span a range of coordination challenges. Corridor stresses spatial coordination under physical constraints, tag requires strategic role assignment, the two SMAX tasks test combat coordination under numerical disadvantage at different scales (8 and 27 agents), and the MABrax tasks test continuous multi-agent locomotion with varying degrees of agent heterogeneity.

\subsubsection{Baselines}

We compare against eight methods: \textbf{IPPO}~\cite{dewitt2020independent} (independent PPO without parameter sharing, MABrax only), \textbf{MAPPO}~\cite{yu2022mappo} (MAPPO without intrinsic reward, MPE and SMAX), \textbf{Linear} (fixed $\beta = 0.1$, tuned), \textbf{COIN}~\cite{li2023coin} (counterfactual intrinsic motivation), \textbf{Lagrangian} (Lagrange multiplier $\lambda$ balances intrinsic and extrinsic advantages as $A^{\intr} + \lambda A^{\ext}$, with $\lambda$ updated via dual gradient descent on mean extrinsic advantage), \textbf{MAVEN}~\cite{mahajan2019maven} (multi-agent variational exploration with latent space coordination), \textbf{RCB-only} (our RCB schedule without RSQ, i.e., $\lambda = 0$), and \textbf{RSQ-only} (per-agent RSQ modulation without return-conditioned adaptation).

\subsubsection{Implementation}
\label{subsec:implementation}

For MPE and SMAX, all methods use MAPPO~\cite{yu2022mappo} with GRU-based actors and a shared critic. For MABrax, all methods use IPPO~\cite{dewitt2020independent} with per-agent feed-forward networks and no shared critic. These choices follow standard practice for each domain~\cite{yu2022mappo, rutherford2024jaxmarl}. Intrinsic rewards are computed via per-agent Successor Distance~\cite{myers2024temporal} using a Metric Residual Network encoder. The local state features $x_t^i$ used for SD differ by domain: spatial coordinates $(x, y)$ for MPE ($d_x = 2$), (health, $x$-position, $y$-position, cooldown) for SMAX ($d_x = 4$), and agent-specific features for MABrax ($d_x = 2$ for ant, $d_x = 10$ for ant-ball, $d_x = 3$ for halfcheetah), following the per-agent factorization in Sec.~\ref{subsec:sd}. Full training hyperparameters, SD encoder configuration, and per-domain details are provided in \ref{app:implementation}. All experiments use 10 random seeds per method on NVIDIA RTX A5000 GPUs. The framework is applied identically across all settings.

\subsection{Main Results}
\label{subsec:exp_baselines}

\begin{figure*}[t]
    \centering
    \begin{subfigure}[t]{0.48\textwidth}
        \centering
        \includegraphics[width=\textwidth]{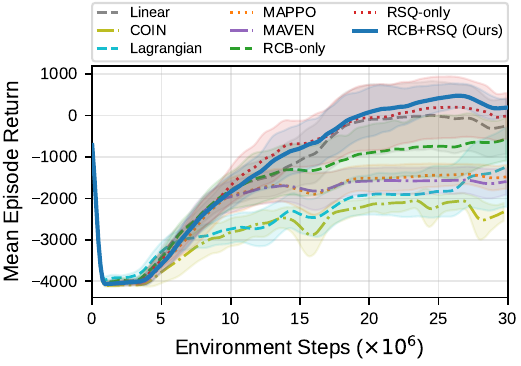}
        \captionsetup{margin={8mm,0mm}}
        \caption{MPE-corridor (8 agents)}
        \label{fig:lc_corridor}
    \end{subfigure}
    \hfill
    \begin{subfigure}[t]{0.48\textwidth}
        \centering
        \includegraphics[width=\textwidth]{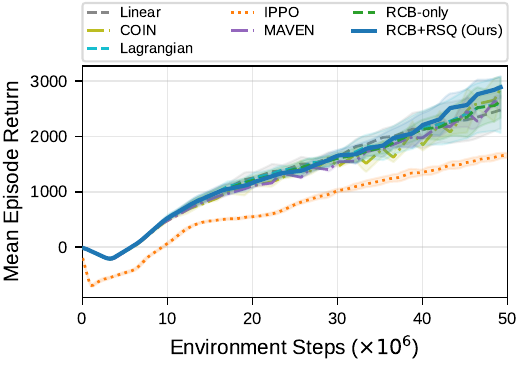}
        \captionsetup{margin={8mm,0mm}}
        \caption{MABrax ant\_4x2 (4 agents, 2 actuators)}
        \label{fig:lc_ant}
    \end{subfigure}
    \caption{Learning curves on representative discrete (left) and continuous (right) multi-agent tasks. Shaded regions indicate $\pm 1$ standard deviation across seeds. On corridor, our method is the only tested method to achieve positive final returns. On ant\_4x2, our method achieves the highest final return among all baselines.}
    \label{fig:learning_curves}
\end{figure*}

\begin{table*}[t]
    \centering
    \caption{Mean episode return ($\pm$ std) across seven environments. \textbf{Bold} marks the highest mean and $^\dagger$ marks a statistical tie with the best ($p > 0.05$, Welch's $t$-test). All methods share the same backbone per domain, namely MAPPO for MPE/SMAX and IPPO for MABrax (Sec.~\ref{subsec:implementation}). Linear uses a fixed exploration weight $\beta = 0.1$. Dashes indicate methods not applicable to that domain. 10 seeds per method.}
    \label{tab:all_results}
    \footnotesize
    \setlength{\tabcolsep}{3.0pt}
    \begin{tabular}{l cc cc ccc}
        \toprule
        & \multicolumn{2}{c}{\textit{MPE (discrete)}} & \multicolumn{2}{c}{\textit{SMAX}} & \multicolumn{3}{c}{\textit{MABrax (continuous)}} \\
        \cmidrule(lr){2-3} \cmidrule(lr){4-5} \cmidrule(lr){6-8}
        Method & Corridor & Tag & 3s5z & 27m & ant\_4x2 & ant\_ball & hc\_6x1 \\
        \midrule
        IPPO & --- & --- & --- & --- & $1659^{(\pm 41)}$ & $830^{(\pm 158)}$ & $2720^{(\pm 114)}$ \\
        MAPPO & $-1477^{(\pm 288)}$ & $+3622^{(\pm 565)}$ & $0.542^{(\pm 0.009)}$ & $0.430^{(\pm 0.015)}$ & --- & --- & --- \\
        RSQ-only & $-20^{(\pm 554)\dagger}$ & $+3823^{(\pm 714)}$ & $0.476^{(\pm 0.033)}$ & $0.448^{(\pm 0.007)\dagger}$ & $2577^{(\pm 368)}$ & $1342^{(\pm 159)}$ & $2589^{(\pm 170)}$ \\
        COIN & $-2729^{(\pm 388)}$ & $+216^{(\pm 663)}$ & $0.193^{(\pm 0.017)}$ & $0.120^{(\pm 0.034)}$ & $2724^{(\pm 427)\dagger}$ & $1702^{(\pm 173)\dagger}$ & $1678^{(\pm 476)}$ \\
        MAVEN & $-1584^{(\pm 432)}$ & $+2329^{(\pm 374)}$ & $0.450^{(\pm 0.094)}$ & $0.238^{(\pm 0.039)}$ & $2282^{(\pm 270)}$ & $1498^{(\pm 128)}$ & $2398^{(\pm 103)}$ \\
        Lagrangian & $-1238^{(\pm 1019)}$ & $+2400^{(\pm 180)}$ & $0.200^{(\pm 0.056)}$ & $0.258^{(\pm 0.065)}$ & $2623^{(\pm 427)\dagger}$ & $1697^{(\pm 292)\dagger}$ & $\mathbf{2874}^{(\pm 141)}$ \\
        RCB-only & $-552^{(\pm 556)}$ & $+3392^{(\pm 584)}$ & $0.582^{(\pm 0.019)\dagger}$ & $0.141^{(\pm 0.021)}$ & $2583^{(\pm 351)}$ & $1690^{(\pm 231)\dagger}$ & $2762^{(\pm 168)\dagger}$ \\
        Linear & $-259^{(\pm 708)\dagger}$ & $+3542^{(\pm 410)}$ & $0.585^{(\pm 0.013)\dagger}$ & $\mathbf{0.449}^{(\pm 0.003)}$ & $2509^{(\pm 341)}$ & $1540^{(\pm 182)}$ & $2746^{(\pm 114)}$ \\
        \textbf{Ours} & $\mathbf{+190}^{(\pm 224)}$ & $\mathbf{+4930}^{(\pm 404)}$ & $\mathbf{0.591}^{(\pm 0.014)}$ & $0.447^{(\pm 0.005)\dagger}$ & $\mathbf{2940}^{(\pm 312)}$ & $\mathbf{1846}^{(\pm 261)}$ & $2870^{(\pm 113)\dagger}$ \\
        \bottomrule
    \end{tabular}
    \smallskip
\end{table*}

We assess pairwise significance using Welch's $t$-test~\cite{welch1947generalization} with $p < 0.05$ as the threshold, marking statistical ties ($p > 0.05$) with $\dagger$ in Table~\ref{tab:all_results}.

\subsubsection{Discrete Control (MPE)}

On MPE-corridor, which demands tight spatial coordination, our method achieves the highest mean return with the lowest variance (Table~\ref{tab:all_results}). The primary distinction over Linear (second-best) is \emph{reliability}, with our method achieving $3.2\times$~lower standard deviation. Every other exploration method ends with strongly negative mean return on this task.

On MPE-tag (Fig.~\ref{fig:lc_tag}), our method achieves the highest mean return (Table~\ref{tab:all_results}). RSQ-only and Linear are the closest competitors but trail our method in mean return while exhibiting larger run-to-run variance.

\subsubsection{Large-Scale Coordination (SMAX)}

\begin{figure*}[t]
    \centering
    \begin{subfigure}[t]{0.48\textwidth}
        \centering
        \includegraphics[width=\textwidth]{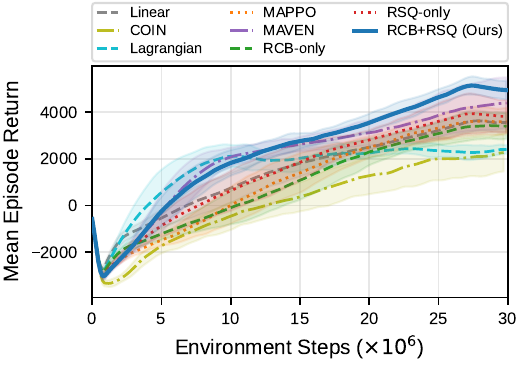}
        \caption{MPE-tag (6 predators, 2 prey)}
        \label{fig:lc_tag}
    \end{subfigure}
    \hfill
    \begin{subfigure}[t]{0.48\textwidth}
        \centering
        \includegraphics[width=\textwidth]{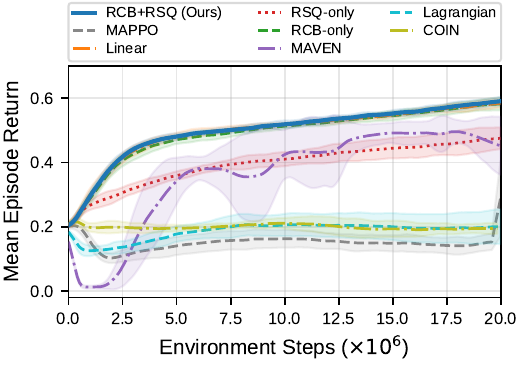}
        \caption{SMAX-3s5z (8 vs.\ 9)}
        \label{fig:lc_3s5z}
    \end{subfigure}
    \caption{Learning curves on MPE-tag (left) and SMAX-3s5z (right). Shaded regions indicate $\pm 1$ standard deviation. On tag, our method converges to the highest return. On SMAX-3s5z, three baselines collapse below $0.2$ while our method maintains the top position.}
    \label{fig:lc_tag_smax}
\end{figure*}

\begin{figure}[!htb]
    \centering
    \includegraphics[width=0.48\textwidth]{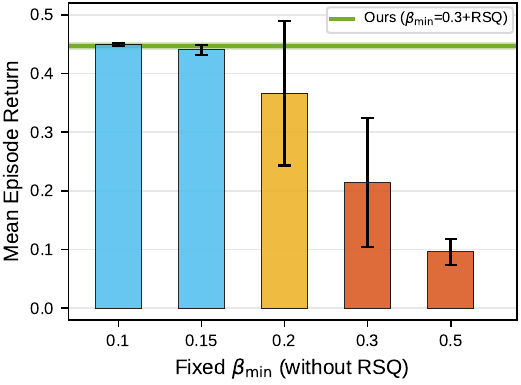}
    \caption{$\beta_{\min}$ sensitivity without RSQ on SMAX-27m. Without RSQ, increasing fixed $\beta_{\min}$ beyond $0.15$ drives the mean return below $0.2$ at $\beta_{\min} = 0.3$ and to near zero at $\beta_{\min} = 0.5$. With RSQ, Ours operates safely at the same $\beta_{\min} = 0.3$.}
    \label{fig:beta_sensitivity}
\end{figure}

On SMAX-3s5z (Table~\ref{tab:all_results}, Fig.~\ref{fig:lc_3s5z}), our method achieves the highest mean return ($0.591$), with Linear and RCB-only statistically tied ($p > 0.10$). The key contrast lies in stability, with COIN and Lagrangian both falling below $0.2$ and MAVEN reaching intermediate performance with $5.5\times$ higher variance than Ours. This task is also the only benchmark with genuinely heterogeneous agent types (stalkers vs.\ zealots), and Our per-agent RSQ modulation naturally adapts to the resulting differences in intrinsic reward dynamics across unit types.

On SMAX-27m (Table~\ref{tab:all_results}), our method and Linear are statistically indistinguishable at convergence ($p = 0.29$), while all other exploration baselines fall below $0.26$. The key distinction is \emph{exploration safety} (Fig.~\ref{fig:beta_sensitivity}), where fixed $\beta = 0.3$ without RSQ drives the mean return below $0.2$, whereas our method operates safely at the same $\beta_{\min} = 0.3$ by selectively attenuating noisy agents.\footnote{The $\beta$-sensitivity data also confirms the contraction condition~\eqref{eq:contraction} empirically: on SMAX ($\kappa{=}0.01$, $\beta_{\max}{-}\beta_{\min}{=}0.2$), the product is $5 \times 10^{-5} \ll 1$.}

\subsubsection{Continuous Control (MABrax)}
\label{subsec:mabrax}

On ant\_4x2, our method achieves the highest return with COIN and Lagrangian statistically tied (Table~\ref{tab:all_results}). On ant\_ball, our method leads with a $9.2\%$ improvement over RCB-only, confirming that per-agent RSQ modulation adds value beyond global scheduling. IPPO collapses on ant\_ball, indicating that independent learning cannot solve this tightly coupled locomotion-balance task.

On halfcheetah, our method and Lagrangian share the top position ($p = 0.95$), consistent with the homogeneous agent roles that limit RSQ differentiation. Across all three continuous tasks, our method delivers consistently strong performance, indicating that per-agent signal quality awareness combined with return-conditioned scheduling is well suited to multi-agent locomotion.

\subsection{Ablation Studies}
\label{subsec:ablation}

\subsubsection{Component Synergy}

Neither RCB-only nor RSQ-only matches the combined framework on corridor. RCB-only provides a global schedule but lacks per-agent signal awareness, so agents with unreliable intrinsic rewards destabilize the policy. RSQ-only improves substantially over unmodulated baselines (corridor mean $-20 \pm 554$ vs.\ MAPPO's $-1477$) but applies a fixed high exploration intensity throughout training, leaving residual instability on corridor and a $22\%$ lower return on tag ($+3823$ vs.\ $+4930$).

On SMAX-27m, RSQ-only achieves $0.448 \pm 0.007$ ($p = 0.72$ vs.\ Ours), suggesting that per-agent modulation is the dominant component. The narrow return range ($0$--$0.45$ vs.\ $R_{\mathrm{target}} = 3.0$) keeps $\beta$ nearly constant, limiting RCB's contribution. On wide-range tasks (corridor, tag, MABrax), both components contribute actively. The combined framework achieves the highest or statistically tied return on all seven environments.

\subsubsection{RSQ Sensitivity (\texorpdfstring{$\lambda$}{lambda})}

\begin{table}[!t]
    \centering
    \caption{$\lambda$ sensitivity on MPE-corridor ($\lambda=0$ is RCB-only). Best in \textbf{bold}.}
    \label{tab:sensitivity_lambda}
    \footnotesize
    \begin{tabular}{ccc}
        \toprule
        $\lambda$ & Mean & Std \\
        \midrule
        $0.0$ & $-551.5$ & $555.7$ \\
        $0.5$ & $-143.0$ & $536.2$ \\
        $1.0$ & $+27.2$ & $430.0$ \\
        $2.0$ & $-82.7$ & $576.7$ \\
        $\mathbf{3.0}$ & $\mathbf{+190.1}$ & $\mathbf{224.0}$ \\
        \bottomrule
    \end{tabular}
\end{table}

\begin{table}[!t]
    \centering
    \caption{$R_{\mathrm{target}}$ sensitivity on MPE-corridor. Best in \textbf{bold}.}
    \label{tab:sensitivity_rtarget}
    \footnotesize
    \begin{tabular}{ccc}
        \toprule
        $R_{\mathrm{target}}$ & Mean Return & Std \\
        \midrule
        $50$ & $-213.8$ & $502.7$ \\
        $100$ & $-229.3$ & $599.2$ \\
        $200$ & $+151.2$ & $314.3$ \\
        $400$ & $+190.1$ & $224.0$ \\
        $600$ & $+234.5$ & $299.6$ \\
        $800$ & $+105.2$ & $144.5$ \\
        $1000$ & $+198.0$ & $290.5$ \\
        \bottomrule
    \end{tabular}
\end{table}

Table~\ref{tab:sensitivity_lambda} shows that any RSQ modulation ($\lambda > 0$) improves both mean return and stability over the unmodulated baseline ($\lambda = 0$). The best result occurs at $\lambda = 3.0$ with $2.5\times$~lower standard deviation.

\subsubsection{\texorpdfstring{$R_{\mathrm{target}}$}{R\_target} Sensitivity}
\label{subsec:rtarget_sensitivity}

As noted in Remark~\ref{rem:rtarget}, the sigmoid's 5\%--95\% transition spans $\Delta R \approx 5.89/\kappa$ return units ($\approx 590$ for $\kappa = 0.01$), so shifting $R_{\mathrm{target}}$ by hundreds of units barely changes the $\beta$ trajectory. Table~\ref{tab:sensitivity_rtarget} confirms this, with all values in $[200, 1000]$ yielding positive mean return and only extreme values ($R_{\mathrm{target}} \leq 100$) destabilizing training.

\subsection{RSQ Dynamics Analysis}

\begin{figure*}[!t]
    \centering
    \includegraphics[width=0.78\textwidth]{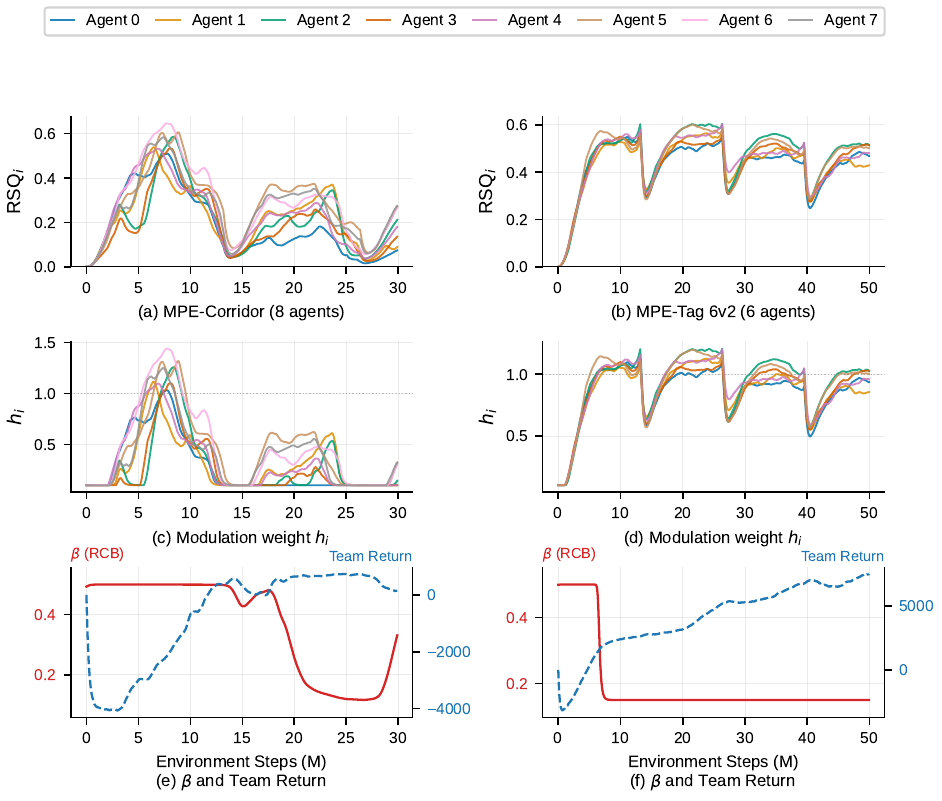}
    \caption{Per-agent RSQ dynamics during training on MPE-corridor (left) and MPE-tag (right). (a,b) Per-agent RSQ values differentiate over training, with corridor exhibiting wider inter-agent spread than tag. (c,d) Modulation weights $h_i$ reflect RSQ differences: agents with low RSQ are attenuated toward $h_{\min}$, concentrating exploration budget on reliable agents. (e,f) The global exploration intensity $\beta$ (RCB) adapts to team performance, decreasing as returns improve. The interaction between per-agent RSQ modulation and global RCB adaptation produces the corridor reliability property observed in Table~\ref{tab:all_results}.}
    \label{fig:rsq_dynamics}
\end{figure*}

Fig.~\ref{fig:rsq_dynamics} visualizes the per-agent RSQ dynamics during a representative training run. On corridor, agents exhibit substantial RSQ differentiation early in training (steps 5M--15M). Open-space agents produce stable SD values across the batch (high RSQ), while bottleneck agents experience high variance due to stochastic multi-agent interactions (low RSQ, lower modulation weights). As agents learn coordinated passage, RSQ values converge as SD intrinsic rewards decline with exploration saturation, but the RCB schedule has already reduced $\beta$ to near $\beta_{\min}$ at this stage. On tag, all agents achieve similar RSQ, consistent with symmetric predator roles. The global $\beta$ responds to team progress: on corridor it remains high until step 15M, while on tag it drops early as returns improve quickly. On halfcheetah, homogeneous roles produce near-uniform RSQ, consistent with the tied performance against Lagrangian in Table~\ref{tab:all_results}.
\FloatBarrier
\section{Conclusions and Future Work}
\label{sec:conclusion}

We presented a framework that resolves the safety-performance dilemma in multi-agent exploration through two complementary mechanisms. Return-conditioned scheduling (RCB) adapts global intensity with provable convergence (Theorem~\ref{thm:rcb_convergence}), and per-agent RSQ modulation concentrates the exploration budget on agents with reliable intrinsic reward signals while preserving the ordering of the information-optimal allocation (Proposition~\ref{thm:ordering}). Across seven cooperative benchmarks spanning discrete, large-scale, and continuous domains, our method achieves top-tier returns.

\textbf{Limitations.}
RSQ modulation requires that the intrinsic reward method produces per-agent signal quality variation (Assumption~\ref{assum:quality_variation}). In our experiments, only Successor Distance satisfies this requirement, while entropy, RND, and count-based rewards yield near-uniform RSQ across agents (\ref{app:intrinsic_source}). Additionally, the current framework assumes cooperative settings with a shared team reward, limiting direct applicability to competitive or mixed-motive scenarios.

\textbf{Future Work.}
Extending RSQ to work with a broader family of intrinsic reward methods is an important open problem. Mixed-motive and competitive settings where the team reward assumption no longer holds, as well as dynamic team compositions where agents join or leave mid-episode, would further test the adaptability of per-agent RSQ allocation. Combining RSQ with topology-aware intrinsic reward designs that exploit the communication graph structure may enhance signal quality differentiation in large-scale teams.

\section*{Acknowledgment}
Funding information will be provided upon acceptance.

\section*{CRediT authorship contribution statement}
\textbf{Dahyun Oh}: Conceptualization, Methodology, Software, Validation, Formal analysis, Investigation, Data curation, Writing -- original draft, Visualization. \textbf{Minhyuk Yoon}: Developing mathematical formulations including the proofs. \textbf{H. Jin Kim}: Supervision, Funding acquisition, Writing -- review \& editing.

\section*{Data availability}
The source code is publicly available at \url{https://github.com/DH-O/RRS}.

\appendix

\section{Proof of Theorem~\ref{thm:rcb_convergence}}
\label{app:rcb_convergence}

The proof proceeds in three steps: (1) establishing that the RCB map is a contraction, (2) bounding the deterministic tracking error, and (3) bounding the stochastic residual.

\textbf{Step 1: Contraction.}
Define the RCB map as a function of the EMA return:
\begin{equation*}
    g(R_{\mathrm{ema}}) \;=\; \beta_{\min} + (\beta_{\max}-\beta_{\min})\,\sigma\!\bigl(\kappa(R_{\mathrm{target}} - R_{\mathrm{ema}})\bigr),
\end{equation*}
so that $\beta^{(k)} = g(R_{\mathrm{ema}}^{(k)})$. The composed one-step map on $\beta$ is $\Phi(\beta) = g\bigl(\bar{R}(\beta)\bigr) = \beta_{\min} + (\beta_{\max}-\beta_{\min})\,\sigma(\kappa(R_{\mathrm{target}} - \bar{R}(\beta)))$. Since $\sigma \in (0,1)$, $\Phi$ maps $[\beta_{\min}, \beta_{\max}]$ to itself. By the chain rule,
\begin{equation*}
    \Phi'(\beta) = -(\beta_{\max}-\beta_{\min}) \cdot \kappa \cdot \sigma'(\cdot) \cdot \bar{R}'(\beta).
\end{equation*}
Since $\sigma'(x) = \sigma(x)(1-\sigma(x)) \leq 1/4$, we have
\begin{equation*}
    L_\Phi \;\triangleq\; \sup_\beta |\Phi'(\beta)| \;\leq\; \kappa(\beta_{\max}-\beta_{\min}) \cdot \sup_\beta|\bar{R}'(\beta)| \cdot \tfrac{1}{4}.
\end{equation*}
The contraction condition~\eqref{eq:contraction} states that the right-hand side is strictly less than 1, which gives $L_\Phi < 1$. By Banach's fixed-point theorem~\cite{granas2003fixed}, $\Phi$ has a unique fixed point $\beta^* = \Phi(\beta^*)$, with corresponding equilibrium return $R^* = \bar{R}(\beta^*)$.

\textbf{Step 2: Deterministic convergence.}
Let $e^{(k)} = R_{\mathrm{ema}}^{(k)} - R^*$. From the EMA update~\eqref{eq:return_ema} with $\xi^{(k)} = 0$:
\begin{align*}
    e^{(k+1)} &= \alpha_R\bigl(\bar{R}(\beta^{(k)}) - \bar{R}(\beta^*)\bigr) + (1-\alpha_R)\,e^{(k)}.
\end{align*}
We bound each term. For the first term, by the mean value theorem, $|\bar{R}(\beta^{(k)}) - \bar{R}(\beta^*)| \leq L_R|\beta^{(k)} - \beta^*|$ where $L_R = \sup_\beta|\bar{R}'(\beta)|$. Since $\beta^{(k)} = g(R_{\mathrm{ema}}^{(k)})$ and $\beta^* = g(R^*)$, the Lipschitz constant of $g$ is $|g'(R_{\mathrm{ema}})| = (\beta_{\max}-\beta_{\min})\kappa\,\sigma'(\cdot) \leq (\beta_{\max}-\beta_{\min})\kappa/4$, giving
\begin{equation*}
    |\beta^{(k)} - \beta^*| \;=\; |g(R_{\mathrm{ema}}^{(k)}) - g(R^*)| \;\leq\; \tfrac{(\beta_{\max}-\beta_{\min})\kappa}{4}\,|e^{(k)}|.
\end{equation*}
Combining: $\alpha_R|\bar{R}(\beta^{(k)}) - \bar{R}(\beta^*)| \leq \alpha_R L_R \cdot \frac{(\beta_{\max}-\beta_{\min})\kappa}{4}\,|e^{(k)}| = \alpha_R L_\Phi\,|e^{(k)}|$, where the last equality uses $L_\Phi = L_R \cdot (\beta_{\max}-\beta_{\min})\kappa/4$ from Step~1. Therefore,
\begin{equation*}
    |e^{(k+1)}| \;\leq\; \bigl[\alpha_R L_\Phi + (1-\alpha_R)\bigr]\,|e^{(k)}| \;=\; \rho\,|e^{(k)}|,
\end{equation*}
where $\rho = 1 - \alpha_R(1-L_\Phi)$. Since $\alpha_R \in (0,1)$ and $L_\Phi < 1$ (from the contraction condition), we have $\rho \in (0,1)$, so the error contracts geometrically.

\textbf{Step 3: Stochastic convergence.}
With noise $\xi^{(k)}$ (Assumption~\ref{assum:return_response}), the EMA update becomes
\begin{equation*}
    e^{(k+1)} = \underbrace{\alpha_R\bigl(\bar{R}(\beta^{(k)}) - \bar{R}(\beta^*)\bigr) + (1-\alpha_R)\,e^{(k)}}_{D^{(k)}} + \alpha_R\xi^{(k)},
\end{equation*}
where $D^{(k)}$ is the deterministic part satisfying $|D^{(k)}| \leq \rho|e^{(k)}|$ from Step~2.

Squaring: $|e^{(k+1)}|^2 = |D^{(k)}|^2 + 2\alpha_R D^{(k)}\xi^{(k)} + \alpha_R^2|\xi^{(k)}|^2$. Taking expectations, the cross term $\E[D^{(k)}\xi^{(k)}] = 0$ because $D^{(k)}$ depends only on $R_{\mathrm{ema}}^{(k)}$ (and hence on $\xi^{(0)}, \ldots, \xi^{(k-1)}$), while $\xi^{(k)}$ is the noise at iteration $k$, which is zero-mean and independent of prior iterations by Assumption~\ref{assum:return_response}. This gives:
\begin{equation*}
    \E\bigl[|e^{(k+1)}|^2\bigr] \;\leq\; \rho^2\,\E\bigl[|e^{(k)}|^2\bigr] + \alpha_R^2\,\sigma_\xi^2.
\end{equation*}
This is a linear recursion $a_{k+1} \leq \rho^2 a_k + \alpha_R^2\sigma_\xi^2$ with $\rho^2 < 1$. Unrolling gives $a_k \leq \rho^{2k} a_0 + \alpha_R^2\sigma_\xi^2/(1-\rho^2)$. As $k \to \infty$, the first term vanishes and
\begin{equation*}
    \lim_{k\to\infty} \E\bigl[|e^{(k)}|^2\bigr] \;\leq\; \frac{\alpha_R^2\sigma_\xi^2}{1-\rho^2}.
\end{equation*}
To connect with the $O$-notation in Theorem~\ref{thm:rcb_convergence}: since $\rho = 1 - \alpha_R(1-L_\Phi)$, we have $1-\rho^2 = (1-\rho)(1+\rho) = \alpha_R(1-L_\Phi)(1+\rho) \geq \alpha_R(1-L_\Phi)$ (because $\rho \geq 0$), so
\begin{equation*}
    \frac{\alpha_R^2\sigma_\xi^2}{1-\rho^2} \;\leq\; \frac{\alpha_R\,\sigma_\xi^2}{1-L_\Phi} \;=\; O\!\left(\frac{\alpha_R\,\sigma_\xi^2}{1-L_\Phi}\right).
\end{equation*}
\hfill$\square$

\section{Proof of Proposition~\ref{thm:ordering} (Ordering Preservation)}
\label{app:ordering_proof}

The water-filling formula~\eqref{eq:water_filling} is standard textbook material~\cite[Ch.~10]{cover2006elements}, and we omit its derivation. Our claim of interest is that the affine RSQ modulation preserves the SNR ordering it induces.

Write $h_i = (h \circ \rsq)(\mathrm{SNR}_i)$, with $\rsq(\mathrm{SNR}) = \mathrm{SNR}/(1+\mathrm{SNR})$ and $h(\rsq) = \mathrm{clip}(1 + \lambda(\rsq - \rsq_{\mathrm{ref}}), h_{\min}, h_{\max})$. The map $\rsq$ has derivative $1/(1+\mathrm{SNR})^2 > 0$ and is therefore strictly increasing on $(0, \infty)$. The clipped affine $h$ has slope $\lambda > 0$ in the unsaturated region $(h_{\min}, h_{\max})$ and slope $0$ when clipped, hence is non-decreasing. The composition $h \circ \rsq$ is therefore non-decreasing in $\mathrm{SNR}$, giving Part~(ii). Part~(i) follows directly from~\eqref{eq:water_filling}, since $1/\mathrm{SNR}_i$ is decreasing in $\mathrm{SNR}_i$. For Part~(iii), if $h_i, h_j \in (h_{\min}, h_{\max})$ then $h$ is the unclipped affine on both, so $h_i - h_j = \lambda(\rsq_i - \rsq_j)$, which is strictly positive whenever $\mathrm{SNR}_i > \mathrm{SNR}_j$ by the strict monotonicity of $\rsq$.

\section{Exact Water-Filling vs.\ Affine RSQ Modulation}
\label{app:wf_comparison}

The methodology (Sec.~\ref{subsec:rsq}) argues that the affine RSQ modulation is preferable to exact water-filling because the latter permanently suppresses agents whose estimated SNR falls below the water level. We validate this claim empirically on corridor (10 seeds each, identical hyperparameters except the allocation mechanism).

\begin{table}[t]
    \centering
    \caption{Exact water-filling vs.\ affine RSQ modulation on corridor (10 seeds). Exact water-filling exhibits 2.5$\times$ higher variance and a strongly degraded mean return because noisy SNR estimates trigger permanent agent suppression ($\beta_i = 0$), destabilizing team coordination.}
    \label{tab:wf_comparison}
    \begin{tabular}{lcc}
        \toprule
        Allocation & Mean Return & Std \\
        \midrule
        Exact water-filling & $-55$ & $560$ \\
        Affine RSQ (Ours) & $\mathbf{+190}$ & $\mathbf{224}$ \\
        \bottomrule
    \end{tabular}
\end{table}

Table~\ref{tab:wf_comparison} confirms that the affine modulation substantially outperforms exact water-filling, which suffers from irreversible agent suppression when noisy SNR estimates push agents below the water level.

\section{Validating the SD Requirement (Assumption~\ref{assum:quality_variation})}
\label{app:intrinsic_source}

To empirically validate Assumption~\ref{assum:quality_variation} and Insight~\ref{insight:quasimetric}, we replace Successor Distance with three alternative intrinsic reward methods on MPE-corridor: policy entropy, Random Network Distillation (RND)~\cite{burda2018exploration}, and visitation count. For each source, we apply the RSQ-modulated configuration over 5 seeds.

\begin{table}[t]
    \centering
    \caption{MPE-corridor performance by intrinsic reward source under RSQ modulation. SD uses 10-seed data from the main experiments (Table~\ref{tab:all_results}). Entropy, RND, and count use 5 seeds each. Only Successor Distance produces sufficient inter-agent RSQ variation for effective modulation. Entropy, RND, and count yield uniform $h_i \approx 1$ across agents, rendering RSQ modulation ineffective regardless of $\lambda$.}
    \label{tab:intrinsic_source}
    \footnotesize
    \begin{tabular}{llcc}
        \toprule
        Intrinsic & Modulation & Mean & Std \\
        \midrule
        SD (Ours) & RSQ ($\lambda = 3.0$) & $\mathbf{+190}$ & $224$ \\
        \midrule
        Entropy & RSQ ($\lambda = 1.5$) & $-1297$ & $275$ \\
        RND & RSQ ($\lambda = 1.5$) & $-1626$ & $51$ \\
        Count & RSQ ($\lambda = 1.5$) & $-1097$ & $936$ \\
        \bottomrule
    \end{tabular}
\end{table}

Table~\ref{tab:intrinsic_source} confirms that only SD produces sufficient inter-agent RSQ variation for effective modulation.

\section{Implementation Details}
\label{app:implementation}

Table~\ref{tab:all_hyper} summarizes all training and SD encoder hyperparameters. All baselines were tuned with comparable effort, selecting hyperparameters from 2--3 candidates on a single seed per environment, matching the tuning protocol used for our method (Table~\ref{tab:hyperparams} caption).

\begin{table}[ht]
    \centering
    \caption{Backbone training hyperparameters by domain.}
    \label{tab:all_hyper}
    \footnotesize
    \begin{tabular}{lccc}
        \toprule
        Parameter & MPE & MABrax & SMAX \\
        \midrule
        backbone & MAPPO & IPPO & MAPPO \\
        shared critic & yes & no & yes \\
        parameter sharing & no & no & no \\
        network type & GRU & FF & GRU \\
        hidden dim & 64 & 64 / 128$^\dagger$ & 64 \\
        GRU dim & 256 & --- & 128 \\
        FC layers & 2 & 2 & 2 \\
        actor learning rate & $1 \times 10^{-3}$$^\S$ & $3 \times 10^{-4}$ & $5 \times 10^{-4}$ \\
        critic learning rate & $5 \times 10^{-4}$ & $3 \times 10^{-4}$ & $5 \times 10^{-4}$ \\
        LR annealing & none & none & none \\
        discount $\gamma$ & 0.99 & 0.99 & 0.99 \\
        GAE $\lambda_{\mathrm{GAE}}$ & 0.95 & 0.95 & 0.95 \\
        clip ratio $\epsilon$ & 0.3 & 0.2 & 0.2 \\
        entropy coeff & 0.015 & 0.01 & 0.01 \\
        max grad norm & 0.6 & 10.0 & 10.0 \\
        update epochs & 10 & 4 & 8 \\
        training steps & $3 \times 10^7$ & $5 \times 10^7$ & $2 \times 10^7$ \\
        parallel envs & 200 & 512 & 128 \\
        rollout length & 256 & 64 & 128 \\
        intrinsic reward scale & 2.0 & 1.0 & 0.01 \\
        warmup rollouts & 0 & 5 & 0 / 20$^\ddagger$ \\
        seeds & 10 & 10 & 10 \\
        \bottomrule
        \multicolumn{4}{l}{\scriptsize $^\dagger$halfcheetah uses 128-dim; ant/ant-ball use 64-dim.} \\
        \multicolumn{4}{l}{\scriptsize $^\S$Tag uses $5 \times 10^{-4}$ actor LR.} \\
        \multicolumn{4}{l}{\scriptsize $^\ddagger$SMAX-27m uses 20 warmup rollouts; SMAX-3s5z uses 0.}
    \end{tabular}
\end{table}

\begin{table}[ht]
    \centering
    \caption{Successor Distance (SD) encoder hyperparameters, shared across all domains.}
    \label{tab:sd_hyper}
    \footnotesize
    \begin{tabular}{lc}
        \toprule
        Parameter & Value \\
        \midrule
        encoder architecture & Metric Residual Network (MRN) \\
        residual blocks & 10 \\
        latent dim & 64 \\
        output dim & 32 \\
        contrastive learning rate & $1 \times 10^{-3}$ \\
        SD discount $\gamma_{\mathrm{SD}}$ & 0.99 \\
        update epochs per rollout & 25 \\
        batch size & 512 \\
        \bottomrule
    \end{tabular}
\end{table}

\begin{table}[ht]
    \centering
    \caption{Per-domain SD input features $x_t^i$, using each agent's own continuous state features and excluding constant indicators (e.g., unit-type bits).}
    \label{tab:sd_features}
    \footnotesize
    \begin{tabular}{lcc}
        \toprule
        Domain & Features & $d_x$ \\
        \midrule
        MPE (corridor, tag) & spatial position $(x, y)$ & 2 \\
        SMAX & health, $x$-pos, $y$-pos, cooldown & 4 \\
        MABrax (ant\_4x2) & joint velocities & 2 \\
        MABrax (ant\_ball) & obs[0:10] (torso + joints) & 10 \\
        MABrax (halfcheetah) & joint velocities & 3 \\
        \bottomrule
    \end{tabular}
\end{table}

\clearpage

\section*{Declaration of generative AI and AI-assisted technologies in the manuscript preparation process}
During the preparation of this work the authors used Claude (Anthropic) in order to assist with grammar checking and code debugging. The authors reviewed and edited all output and take full responsibility for the content of the published article.

\bibliographystyle{elsarticle-num}
\bibliography{references}

\end{document}